\documentclass[reprint,amsmath,amssymb,aps,twocolumn,pra]{revtex4-2}

\usepackage{dcolumn}
\usepackage{bm}

\usepackage{relsize}
\usepackage[utf8]{inputenc}
\usepackage[english]{babel}
\usepackage[T1]{fontenc}
\usepackage{hyperref}
\usepackage{graphicx}

\usepackage{amsthm}
\usepackage{amsmath} 
\usepackage{amssymb}
 
\newtheorem{theorem}{Theorem}
\newtheorem{corollary}[theorem]{Corollary}
\newtheorem{lemma}[theorem]{Lemma}
\newtheorem{definition}[theorem]{Definition}
\newtheorem{remark}[theorem]{Remark}
\newtheorem{example}[theorem]{Example}

\begin{document}

\preprint{APS/123-QED}

\title{Degenerate Local-dimension-invariant Stabilizer Codes and an Alternative Bound for the Distance Preservation Condition}

\author{Lane G. Gunderman}
 \email{lgunderman@uwaterloo.ca}
\affiliation{%
 The Institute for Quantum Computing, University of Waterloo, Waterloo, Ontario, N2L 3G1, Canada
}%
\affiliation{Department of Physics and Astronomy, University of Waterloo, Waterloo, Ontario, N2L 3G1, Canada}

\date{\today}

\begin{abstract}
One hurdle to performing reliable quantum computations is overcoming noise. One possibility is to reduce the number of particles needing to be protected from noise and instead use systems with more states, so called qudit quantum computers. In this paper we show that codes for these systems can be derived from already known codes, and in particular that degenerate stabilizer codes can have their distance also promised upon sufficiently large local-dimension, as well as a new bound on the local-dimension required to preserve the distance of local-dimension-invariant codes, which is a result which could prove to be useful for error-corrected qudit quantum computers.
\end{abstract}

\maketitle

Having protected quantum information is an essential piece of being able to perform quantum computations. There are a variety of methods to help protect quantum information such as those discussed in \cite{lidar2013quantum}. In this work we focus on stabilizer codes as they are the quantum analog of classical linear codes. Even with error-correcting codes, having sufficient amounts of protected quantum information to perform useful tasks is still an unresolved challenge. A way to retain a similarly sized computational space while reducing the number of particles that need precise controls is to replace the standard choice of qubits with \textit{qudits}, quantum particles with $q$ levels, also known as local-dimension $q$ \cite{wang2020qudits}. Throughout this work we require $q$ to be a prime so that each nonzero element has a unique multiplicative inverse over $\mathbb{Z}_q$. This restriction can likely be removed, but for simplicity and clarity we only consider this case. Experimental realizations of qudit systems are currently underway \cite{low2020practical,quditlight,kononenko2020characterization,yurtalan2020implementation}, so having more error-correcting codes will aid in protecting such systems.

Prior work on qudit error-correcting codes have at times had challenging restrictions between the parameters of the code \cite{quditgeneral,quditbch,quditmds}, and we've already made progress on reducing this barrier in a prior paper \cite{gunderman2020local}. Our prior work showed the ability to make error-correcting codes that preserved their parameters even upon changing the local-dimension of the system, provided the local-dimension is sufficiently large. Unfortunately the ability to promise the distance of the codes was only shown for non-degenerate codes and with a large local-dimension value required. Beyond this, qudits also have proven connections to foundational aspects of physics \cite{contextuality}. Seeing these potential reasons for using qudits, this work builds off of our prior work to expand the local-dimension-invariant work to the case of degenerate codes as well as providing a roughly quadratic improvement in the size of the local-dimension needed to still promise the distance of the code. With these results the practicality of using this method is improved as well as now providing the option of applying the result to the essential class of degenerate codes, such as quantum versions of low-density parity-check (LDPC) codes.

\section{Definitions for Qudit Stabilizer Codes}

In this section we review some key facts about qudit stabilizer codes. For a more complete guide on qudit stabilizer codes, we recommend \cite{quditgeneral}. The definitions laid out here will be used throughout this work. Let $q$ be the local-dimension of a system, where $q$ is a prime number. We will denote by $\mathbb{Z}_q$ the set $\{0,1,\ldots q-1\}$. When $q=2$ we refer to each register as a qubit, while for any value of $q$ we call each register a qudit. In order to speak more generally and not specify $q$, we will often times refer to each register as a particle instead. We now begin to define the operations for these registers.

\begin{definition}\label{dim}
Generalized Paulis for a particle over $q$ orthogonal levels (local-dimension $q$) are given by:
\begin{equation}
 X_q|j\rangle=|(j+1)\mod q\rangle,\quad Z_q|j\rangle=\omega^j|j\rangle
\end{equation}
with $\omega=e^{2\pi i/q}$, where $j\in\mathbb{Z}_q$. These Paulis form a group, denoted $\mathbb{P}_q$.
\end{definition}

When $q=2$, these are the standard qubit operators $X$ and $Z$, with $Y=iXZ$. This group structure is preserved over tensor products since each of these Paulis has order $q$. A generalized Pauli over $n$ registers is a tensor product of $n$ generalized Pauli group members over a single register. 

A commuting subgroup of generalized Pauli operators with $n-k$ generators, but not including any nontrivial coefficient for the identity operator, is equivalent to a stabilizer code. The number of orthogonal eigenvectors, which form bases called codewords, for these $n-k$ generators is $q^k$. In effect, we have constructed $k$ \textit{logical} particles from the $n$ \textit{physical} particles. If we are to use these subgroups for error-correction purposes then they ought to be able to have some accidental operator occur and still have the codewords be discernible.
We will work under the assumption that errors on distinct particles are independent and we will assume the error model on each qudit is the depolarizing channel.\if{false}, which we define as follows:
\begin{equation}
    \mathcal{E}(\rho) = (1-p)\rho + \frac{p}{q^2-1} \sum_{E \in \mathbb{P}_q \setminus \{I\}} E \rho E^\dagger
\end{equation}\fi
Given this error model we will predominantly be interested in the number of non-identity terms in any error as the exponent of the error term increases with this.

\begin{definition}
The weight of an $n$-qudit Pauli operator is the number of non-identity operators in it.
\end{definition}

\begin{definition}
A stabilizer code, specified by its $n-k$ generators, is characterized by the following set of parameters:
\begin{itemize}
\item $n$: the number of (physical) particles that are used to protect the information.
\item $k$: the number of encoded (logical) particles.
\item $d$: the distance of the code, given by the lowest weight of an undetectable generalized Pauli error. An undetectable generalized Pauli error is an $n$-qudit Pauli operator which commutes with all elements of the stabilizer group, but is not in the group itself.
\end{itemize}
These values are specified for a particular code as $[[n,k,d]]_q$, where $q$ is the local-dimension of the qudits.
\end{definition}

We pause for a moment here to discuss how degenerate codes differ from non-degenerate codes. Degenerate codes are different in the following equivalent ways. Firstly, they may have multiple errors with the same syndrome value and that map to different physical states, but upon recovery still map back to the same logical state. Secondly, degenerate codes may have generators, aside from the identity operator, which have lower weight than the distance of the code. These two differences make degenerate codes markedly different from their non-degenerate counterpart. Degenerate codes, while having these extra nuances, are a crucial class of stabilizer codes as any quantum analog of a low-density parity-check (LDPC) code with high distance will need to be a degenerate code. We will begin our new results by focusing on non-degenerate codes, then move to the degenerate case in Theorem \ref{degen}, however, there are more tools needed before discussing the new results.


Working with tensors of operators can be challenging, and so we make use of the following well-known mapping from these to vectors, following the notation from \cite{gunderman2020local}. This representation is often times called the symplectic representation for the operators, but we use this notation instead to allow for greater flexibility, particularly in specifying the local-dimension of the mapping. This linear algebraic representation will be used for our proofs.

\begin{definition}[$\phi$ representation of a qudit operator]
We define the linear surjective map: 
\begin{equation}
\phi_q: \mathbb{P}_q^n\mapsto \mathbb{Z}_q^{2n}
\end{equation}
which carries an $n$-qudit Pauli in $\mathbb{P}_q^n$ to a $2n$ vector mod $q$, where we define this mapping by:
\begin{equation}
I^{\otimes i-1} X_q^a Z_q^b I^{\otimes n-i} \mapsto \left( 0^{i-1}\ a\ 0^{n-i} \middle\vert 0^{i-1}\ b\ 0^{n-i}\right),
\end{equation}
which puts the power of the $i$-th $X$ operator in the $i$-th position and the power of the $i$-th $Z$ operator in the $(n+i)$-th position of the output vector. This mapping is defined as a homomorphism with: $\phi_q(s_1\circ s_2)=\phi_q(s_1)\oplus \phi_q(s_2)$, where $\oplus$ is component-wise addition mod $q$. We denote the first half of the vector as $\phi_{q,x}$ and the second half as $\phi_{q,z}$.
\end{definition}

We may invert the map to return to the original $n$-qudit Pauli operator with the global phase being undetermined. We make note of a special case of the $\phi$ representation:

\begin{definition}
Let $q$ be the dimension of the initial system. Then we denote by $\phi_\infty$ the mapping:
\begin{equation}
    \phi_\infty:  \mathbb{P}_q^n\mapsto \mathbb{Z}^{2n}
\end{equation}
where no longer are any operations taken $\mod$ some base, but instead carried over the full set of integers.
\end{definition}

The ability to define $\phi_\infty$ as a homomorphism still (and with the same rule) is a portion of the results of \cite{gunderman2020local}. $\phi_q$ is the standard choice for working over $q$ bases, however, our $\phi_\infty$ allows us to avoid being dependent on the local-dimension of our system when working with our code. Formally we will write a code in $\phi_q$, perform some operations, then write it in $\phi_\infty$, then select a new local-dimension $q'$ and use $\phi_{q'}$. We shorten this to write it as $\phi_\infty$, and can later select to write it as $\phi_{q'}$ for some prime $q'$ by taking element-wise $\mod q'$.

The commutator of two operators in this picture is given by the following definition:
\begin{definition}
Let $s_i,s_j$ be two qudit Pauli operators over $q$ bases, then these commute if and only if:
\begin{equation}
\phi_q(s_i)\odot \phi_q(s_j)=0\mod q
\end{equation}
where $\odot$ is the symplectic product, defined by:
\begin{multline}
\phi_q(s_i)\odot \phi_q(s_j)\\ =\oplus_k [\phi_{q,z}(s_j)_k\cdot  \phi_{q,x}(s_i)_k- \phi_{q,x}(s_j)_k \cdot \phi_{q,z}(s_i)_k]
\end{multline}
where $\cdot$ is standard integer multiplication $\mod q$ and $\oplus$ is addition $\mod q$.
\end{definition}

When the commutator of $s_i$ and $s_j$ is not zero, this provides the difference in the number of $X$ operators in $s_i$ that must pass a $Z$ operator in $s_j$ and the number of $Z$ operators in $s_i$ that must pass an $X$ operator in $s_j$ when attempting to switch the order of these two operators.

Before finishing, we make a brief list of some possible operations we can perform on our $\phi$ representation:
\begin{enumerate}
    \item We may perform elementary row operations over $\mathbb{Z}_q$, corresponding to relabelling and composing generators together.
    \item We may swap registers (qudits) in the following ways:
        \begin{enumerate}
            \item We may swap columns $(i,i+n)$ and $(j,j+n)$ for $1\leq i,j\leq n$, corresponding to relabelling qudits.
            \item We may swap columns $i$ and $(-1)\cdot (i+n)$, for $1\leq i\leq n$, corresponding to conjugating by a Hadamard gate on particle $i$ (or Discrete Fourier Transforms in the qudit case \cite{qudit}) thus swapping $X$ and $Z$'s roles on that qudit.
        \end{enumerate}
\end{enumerate}

All of these operations leave the code parameters $n$, $k$, and $d$ alone, but can be used in proofs.

\subsection{Local-dimension-invariant Codes}

In this section we recall the results relating to local-dimension-invariant (LDI) stabilizer codes. These codes answer the question of when we can apply a code from one local-dimension $q$ on a system with a different local-dimension $p$. While an unusual property, a LDI code would permit the importing of smaller local-dimension codes for larger local-dimension systems. Some codes with particular parameters may not be known and so this fills in some of these gaps. Additionally, this framework could potentially provide insights into local-dimension-invariant measurements. Few examples of LDI codes, although not by this name, were known, notable the 5 particle code \cite{5qudit} and the 9 particle code \cite{chau1997correcting}, until the recent work in \cite{gunderman2020local} which showed that all codes can satisfy the commutation requirements, and at least for sufficiently large local-dimension the distance can also be at least preserved. We will review next the primary results from that work.

\begin{definition}
A stabilizer code $S$ is called local-dimension-invariant (LDI) iff:
\begin{equation}
    \phi_\infty(s_i)\odot \phi_\infty(s_j)=0,\quad \forall s_i,s_j\in S.
\end{equation}
\end{definition}

As an example, consider the two qubit code generated by $\langle X\otimes X,Z\otimes Z\rangle$. The symplectic product between the two generators is $2$, so it makes it a valid qubit code, however, $2\mod p\neq 0$ unless $p=2$, so it is not a valid qudit code for $p\neq 2$. If we instead transform the code into one generated by $\langle X\otimes X^{-1},Z\otimes Z\rangle$, then the symplectic product is now $0$, and so it can be used as generators for any choice of local-dimension, and so is an LDI code. The next statement explains that it is always possible to do so:

\begin{theorem}\label{inv}
All stabilizer codes, $S$, can be put into an LDI form. One such method is to put $S$ into canonical form $[I_k\ X_2\ |\ Z_1\ Z_2]$ then transform the code into $[I_k\ X_2\ |\ Z_1+L\ \ Z_2]$, with $L_{ij}=\phi_\infty(s_i)\odot \phi_\infty(s_j)$ when $i>j$ and $0$ otherwise.
\end{theorem}
\if{false}
\begin{proof}
Begin by putting $S$ into canonical form so that:
\begin{equation}
    \phi_q(S)=[I_{n-k}\ \ X_2\ |\ Z_1\ Z_2],
\end{equation}
then the modified code given by:
\begin{equation}
    \phi_\infty(S'):=[I_{n-k}\ \ X_2\ |\ Z_1+L\ \ Z_2],
\end{equation}
where $L_{ij}=\phi_\infty (s_i)\odot \phi_\infty(s_j)$ when $i>j$ and $0$ otherwise. Then $S'$ satisfies the LDI condition and has $\phi_q(S')=\phi_q(S)$.
\end{proof}
\fi
Note that this does not say all codes have a \textit{unique} LDI form, just that there exists one. The proof used is useful as it gives a prescriptive method for turning a code into an LDI form, however, if one does not put the code into canonical form, the code can still be transformed into an LDI form as this process is equivalent to finding solutions to an integer linear program with an abundance of variables. 

As of this point we have merely generated a set of commuting operators that are local-dimension independent. This does not provide for any claims on the distance of the code produced through this method aside from promising that the procedure does not change the distance of the code over the initial local-dimension $q$. For this, we have the following Theorem:
\begin{theorem}
For all primes $p>p^*$, with $p^*$ a cutoff value greater than $q$, the distance of an LDI form of a non-degenerate stabilizer code $[[n,k,d]]_q$ into $p$ bases, $[[n,k,d']]_p$, has $d'\geq d$.
\end{theorem}
There are two caveats to this result, one of which we resolve here, the other of which we provide an improvement on. Let $B$ be the maximal entry in $\phi_\infty(S)$. Firstly, this result is only for the case of non-degenerate codes. We will resolve this with Theorem \ref{degen}. Secondly, the initially proven bound was $p^*=B^{2(d-1)}(2(d-1))^{(d-1)}$, which grows very rapidly. While it was true that all primes below $p^*$ could have their distances checked computationally, this still left a large number of primes to check in most cases. In this work we manage to prove an alternative bound that has a nearly quadratic improvement on the dependency on $B$, as well as a cutoff bound $p_D^*$ in the degenerate case. In the next section we show this alternative cutoff bound, while in the section thereafter the ability to provide a distance promise for degenerate codes is proven.

\section{Alternative Cutoff Bound for the Distance Promise}

While the proof of Theorem \ref{inv} from \cite{gunderman2020local} used $L_{ij}=\phi_\infty(s_i)\odot \phi_\infty(s_j)$ in order to generate a single LDI form, we may generate other LDI forms by altering the added $L$ matrix. We note two of these now: $L^{(+)}$ and $L^{(-)}$.

\begin{definition}
$L^{(+)}$ ($L^{(-)}$) has $L_{ij}^{(+)}$ ($L_{ij}^{(-)}$) is $\phi_\infty(s_i)\odot \phi_\infty(s_j)$ if the symplectic product is greater than zero (less than zero).
\end{definition}

These alternative $L$ matrices each provide a different property. Firstly, using $L^{(+)}$ allows $\phi_\infty(S)$ to have only non-negative entries. There are certain properties that are only generally true for matrices with non-negative entries, so this can perhaps be of use. Additionally, this could be of use for systems formally with countably infinite local-dimension, such as Bosonic systems. Secondly, $L^{(-)}$ permits a slight reduction in the bound for the maximal entry in $\phi_\infty(S)$, as the following Lemma shows:

\begin{lemma}
The maximal entry in $\phi_\infty(S)$, $B$, can be at most $(1+k(q-1))(q-1)$, and generally $B\leq \max_{i,j}|\phi_\infty(s_i)\odot \phi_\infty(s_j)|$.
\end{lemma}

Upon putting the code into canonical form this follows immediately from the definition of $L^{(-)}$ as each entry will be whatever value was already in that location (values in $\mathbb{Z}_q$) minus the absolute value of the inner product, which will be at most an absolute value of the inner product. While this is a small improvement on the value of $B$, since it's the base of an exponential expression this amounts to a larger improvement in the overall cutoff value.

We will now move to proving an alternative bound on the local-dimension needed in order to promise the distance is at least preserved. The first proof of the cutoff bound for the distance promise for LDI codes used random permutations of the entries in $\phi_\infty$. Here we utilize the structure of the symplectic product as well as that of the partitions of the code in terms of its $X$ component and $Z$ component to obtain an alternative bound for all non-degenerate codes. While this bound is looser when $d$ increases, for small $d$ and large $k$ this bound will typically be roughly quadraticly smaller. In particular we will show:

\begin{theorem}\label{improvedbound}
For all primes $p>p^*$ the distance of an LDI representation of a non-degenerate stabilizer code $[[n,k,d]]_q$ over $p$ bases, $[[n,k,d']]_p$, has $d'\geq d$, where we may use as $p^*$ the value:
\begin{equation}
    (B(q-1)(d-1)(1+(d-1)^2(q-1)^{d-1}(d-2)^{(d-2)/2}))^{d-1},
\end{equation}
with $q$ the initial local-dimension, $d$ the distance of the initial code, and $B$ the maximal entry in the $\phi_\infty$ representation of the code.
\end{theorem}

To make claims about the distance of the code we begin by breaking down the set of undetectable errors into two sets. These definitions highlight the subtle possibility of the distance reducing upon changing the local-dimension.

\begin{definition}
An unavoidable error is an error that commutes with all stabilizers and produces the $\vec{0}$ syndrome over the integers.
\end{definition}

These correspond to undetectable errors that would remain undetectable regardless of the number of bases for the code since they always exactly commute under the symplectic inner product with all stabilizer generators--and so all members of the stabilizer group. Since these errors are always undetectable we call them unavoidable errors since changing the number of bases would not allow this code to detect this error.

We also define the other possible kind of undetectable error for a given number of bases, which corresponds to the case where some syndromes are multiples of the number of bases:

\begin{definition}
An artifact error is an error that commutes with all stabilizers but produces at least one syndrome that is only zero modulo the base.
\end{definition}

These are named artifact errors as their undetectability is an artifact of the number of bases selected and could become detectable if a different number of bases were used with this code. Each undetectable error is either an unavoidable error or an artifact error. We utilize this fact to show our theorem.

\begin{proof}

Let us begin with a code with local-dimension $q$ and apply it to a system with local-dimension $p$. The errors for the original code are the vectors in the kernel of $\phi_q$ for the code. These errors are either unavoidable errors or are artifact errors. \if{false} We may rearrange the rows and columns so that the stabilizers and registers that generate these entries that are nonzero multiples of $q$ are the upper left $2d\times 2d$ minor, padding with identities if needed. The factor of 2 occurs due to the number of nonzero entries in $\phi_\infty$ being up to double the weight of the Pauli.\fi The stabilizers that generate these multiples of $q$ entries in the syndrome are members of the null space of the minor formed using the corresponding stabilizers.

Now, consider the extension of the code to $p$ bases. Building up the qudit Pauli operators by weight $j$, we consider the minors of the matrix. These minors of size $2j\times 2j$ can have a nontrivial null space in two possible ways:
\begin{itemize}
    \item If the determinant is 0 over the integers then this is either an unavoidable error or an error whose existence did not occur due to the choice of the number of bases.
    \item If the determinant is not 0 over the integers, but takes the value of some multiple of $p$, then it's $0\mod p$ and so a null space exists.
\end{itemize}
Thus we can only introduce artifact errors to decrease the distance. By bounding the determinant by $p^*$, any choice of $p>p^*$ will ensure that the determinant is a unit in $\mathbb{Z}_p$, and hence have a trivial null space since the matrix is invertible.

We next utilize the structure of the symplectic product more heavily in order to reduce the cutoff local-dimension. Note that for a pair of Paulis in the $\phi$ representation, we may write:
\begin{eqnarray}
    \phi(s_1)\odot \phi(s_2)&=&\phi(s_1)\begin{bmatrix} 0 & -I_n\\ I_n & 0 \end{bmatrix} \phi(s_2)^T\\
    &:=&\phi(s_1)g \phi(s_2)^T
\end{eqnarray}
and so we may consider the commutation for the generators with some Pauli $u$ as being given by $\bigoplus_{i=1}^{n-k} (\phi(s_i)g)\phi(u)^T$. This removes the distinction between the two components and allows the symplectic product to act like the normal matrix-vector product. Now, notice that for any Pauli weight $j$ operator, we will have up to $j$ nonzero entries in the $X$ component of the $\phi$ representation and up to $j$ nonzero entries in the $Z$ component. This means that up to $j$ columns in each component will be involved in any commutator.

Next note that to ensure that an artifact error is not induced it suffices to ensure that there is a nontrivial kernel, induced by the local-dimension choice, which is ensured so long as any $2(d-1)\times 2(d-1)$ minor does not have a determinant which is congruent to the local-dimension. This can be promised by requiring the local-dimension to be larger than the largest possible determinant for such a matrix. Since there will be at most $j$ nonzero entries in each component it suffices to consider $j$ columns from each component and subsets of $2j$ rows of this.

From this reduction, we need only ensure that the local-dimension is larger than the largest possible determinant for this $2j\times 2j$ minor. Let us denote this minor by:
\begin{equation}
    \begin{bmatrix}
    X_1 & Z_1\\
    X_2 & Z_2 
    \end{bmatrix},
\end{equation}
where each block has dimensions $j\times j$. The maximal entries are $q-1$ for $X_1$ and $X_2$, whereas for $Z_1$ and $Z_2$ it is bounded by $B$. We now use the block matrix identity:
\begin{equation}
    det\begin{bmatrix}
    X_1 & Z_1\\
    X_2 & Z_2
    \end{bmatrix}=det(X_1)det(Z_2-X_2X_1^{-1}Z_1).
\end{equation}

Since all entries in $X_1$ are integers and the determinant is, by construction, nonzero, the maximal entry in $X_1^{-1}$ will be at most that of the largest cofactor of $X_1$. The largest cofactor, $\tilde{C}$, will be at most $(q-1)^{d-2}(d-2)^{(d-2)/2}$, as provided by Hadamard's inequality. The largest entry in $Z_2-X_2X_1^{-1}Z_1$ is then upper bounded by $B(1+(q-1)\tilde{C}(d-1)^2)$. From here, we may apply Hadamard's inequality for determinants again using the given entry bounds, using that each block has dimensions up to $(d-1)\times (d-1)$, which provides $p^*=(q-1)^{d-1}(d-1)^{d-1}(B(1+(q-1)\tilde{C}(d-1)^2))^{d-1}$, or alternatively expressed in terms of our fundamental variables as
\begin{equation}
 (B(q-1)(d-1)(1+(d-1)^2(q-1)^{d-1}(d-2)^{(d-2)/2}))^{d-1}.
\end{equation}
In the case of $q=2$ this reduces to $(B(d-1)(1+(d-1)^2(d-2)^{(d-2)/2}))^{d-1}$.

Lastly, when $j=d$, we can either encounter an unavoidable error, in which case the distance of the code is $d$ or we could obtain an artifact error, also causing the distance to be $d$. It is possible that neither of these occur at $j=d$, in which case the distance becomes some $d'$ with $d<d'\leq d^*$, with $d^*$ being the distance of the code over the integers. 
\end{proof}

Before concluding this section, we provide a brief comparison of this bound to the original one of $B^{2(d-1)}(2(d-1))^{(d-1)}$. The new bound only depends on $B^{d-1}$ opposed to the original $B^{2(d-1)}$, which as the bound on $B$ depends on $k$ means that for codes, or code families, with larger $k$ values the new bound can provide a tighter expression. Unfortunately, however, this new bound is doubly-exponential in the distance of the code $d$, having a dependency of roughly $d^{d^2}$ opposed to the prior dependency of $d^d$, so if one is attempting to promise the distance of a code with a larger distance, this new bound is likely to be far less tight. In summary, this alternative bound is not per se better, however, since one may simply use whichever of the bounds is tighter this alternative bound may provide a lower requirement for the local-dimension needed in order to ensure that the distance of the code is at least preserved.

\if{false}
\begin{proof}

\if{false}
Sufficient to have:
\begin{equation}
    det(A\oplus B)\leq p
\end{equation}
but note that $det(X_d\oplus Z_d)=det(X_d)det(Z_d)$, which provides a large improvement. Goes from $2d\times 2d$ matrices to $d\times d$ and maximal entry $B$ to $B$ and $q$.

$((q-1)B)^{d-1}(d-1)^{d-1}$.\fi
\end{proof}
\fi

\section{Degenerate Codes}

Degenerate codes are a uniquely quantum phenomenon, however, are a crucial class of quantum error-correcting codes. For a degenerate quantum error-correcting code we must avoid undetectable errors, but also detectable errors which produce the same syndrome but do not map to the same logical codeword. Any LDPC-like quantum error-correcting code will be degenerate, as, equivalently, a quantum error-correcting code is degenerate if there is some stabilizer group member with lower weight than the distance of the code and by construction one would aim to have a high distance but still $O(1)$ weight for each generator. We show now that a similar distance promise may be made in the degenerate case as was possible in the non-degenerate case.

\begin{theorem}\label{degen}
The distance promise can also be made for degenerate codes.
\end{theorem}

\begin{proof} For the undetectable error case, this follows the same reasoning as the non-degenerate case of this theorem, so we only need to worry about two errors with the same syndrome mapping to different logical states.

Let $u$ and $v$ be two Paulis not in $S$ with weight at most $d-1$. We will prove that upon achieving local-dimension values at least as large as $p_D^*$, a cutoff value, it is impossible to introduce any new degenerate errors. Next, let $\phi_\infty(S)|_u$ and $\phi_\infty(S)|_v$ be the columns in the $\phi_\infty$ representation of the code that are used in determining the syndromes for $u$ and $v$, respectively. Further let $\phi_\infty(S)|_{u\cup v}$ be the union of the columns from $\phi_\infty(S)|_u$ and $\phi_\infty(S)|_v$. This will be up to $4(d-1)$ columns, all of which are not scalar multiples of any other column--so long as $d\geq 3$. If $\phi(u)$ and $\phi(v)$ have the same syndromes then it must be the case that $\phi_\infty(S)|_u$ and $\phi_\infty(S)|_v$ have some linearly dependent set of columns. As we are concerned with errors which are induced as degenerate errors under a change in the local-dimension we will omit the trivial linearly dependent columns (those which are the same columns up to a scalar multiple). Upon removing those columns we now only have $\phi_\infty(S)|_{u\cup v}$, which we aim to avoid having a linearly dependent set of columns over the new local-dimension. This can be avoided by not allowing the introduction of a nontrivial kernel for this matrix. From here the same determinant bound technique may be used. This provides $p_D^*$ as the $2(d-1)$ versions of the bounds:
\begin{multline}
    p_D^*=\min\{B^{4(d-1)}(4(d-1))^{2(d-1)},\\
((q-1)2(d-1)(B(1+(q-1)\tilde{C}(2(d-1))^2))^{2(d-1)}\},
\end{multline}
with $\tilde{C}\leq ((q-1)^2(2d-3))^{d-3/2}$. In totality, we must have $p>\max\{p^*,p_D^*\}$ in order to ensure the distance is preserved.
\end{proof}

\if{false}
For the degenerate case the concern is having two syndromes map to each other that weren't before. This can be avoided if $p>p_d^*$ where $p_d^*=2d(q-1)^2$ (which is the maximal value for a syndrome value (might be $B^2$ instead)). Not sure about this. Basically detectable errors are the only difference here. Is it always the case that $p_d^*<p^*$? If so then degenerate codes are no different from nondegenerate, in terms of proof. Recast as a difference of two Paulis, with a vector congruent to zero over $p$. The $2d\times 2d$ minors used can be independent so this will replace $B$ with $2B$, but otherwise be the same, I think. This suffices but is not strictly needed.

What if the difference is exactly zero. Is this already covered by the original code $q$? I believe so. Need to argue this one too.

No, this vector difference is not working either since this allows for differing vectors, it's not the case that a trivial null space suffices.

$Ax-By=0\mod p\Longleftrightarrow (A-B\tau)x=0\mod p$, with $\tau$ being some permutation (or restricted permutation) matrix? This permutes the rows and columns so that both matrices act on the same space. I think this patches the above.

\begin{lemma}
Reduction of $p^*$, maybe through linear combination argument--can we get $B$ down to $2q$?
\end{lemma}
\fi

This means that just like non-degenerate quantum codes, we may also promise the distance of the code in the degenerate case, albeit with a larger required cutoff value. While this cutoff value is large, it provides some local-dimension value beyond which the distance will be kept and bounds the set of local-dimension values for which the distance must be manually verified. This provides information about when the distance of the code must be preserved, however, if we apply a code over $q$ levels to a system with $p<q$ levels, is there some range of values for $p$ whereby we know that the distance must decrease? In the non-degenerate case, we denoted this by $p^{**}$, which was given by: \begin{equation}
    \sqrt{1+{n\choose t}^{1/((n-k)-t)}},\quad t=\left\lfloor\frac{d-1}{2} \right\rfloor.
\end{equation}
Whenever $p<p^{**}$, it must be the case that the distance of the code must decrease. The expression for $p^{**}$ was derived by using the generalized quantum Hamming bound, which holds for all non-degenerate codes, however, for degenerate codes this bound does not always hold . This means that for a general degenerate code we have the following Lemma:
\begin{lemma}
There is no corresponding $p_D^{**}$ that holds for arbitrary degenerate codes. 
\end{lemma}

While not all degenerate quantum codes obey the generalized quantum Hamming bound, there are certain code families which do \cite{quditgeneral,sarvepalli2010degenerate}. For those code families the exact same expression for $p^{**}$ holds as did for non-degenerate codes.

To ground the prior result, we provide an example next.

\begin{example}
For our example we consider the six qubit code, with parameters $[[6,1,3]]_2$, generated by an extension of the five qubit code. The generators for this code can be given by $\{YIZXYI,ZXIZYI,ZIXYZI,IIIIIX,IZZZZI\}$. In the $\phi_2$ representation this is given by:
\begin{equation}
\setcounter{MaxMatrixCols}{19}
    \begin{bmatrix}
    1 & 0 & 0 & 1 & 1 & 0 & | & 1 & 0 & 1 & 0 & 1 & 0\\
    0 & 1 & 0 & 0 & 1 & 0 & | & 1 & 0 & 0 & 1 & 1 & 0\\
    0 & 0 & 1 & 1 & 0 & 0 & | & 1 & 0 & 0 & 1 & 1 & 0\\
    0 & 0 & 0 & 0 & 0 & 1 & | & 0 & 0 & 0 & 0 & 0 & 0\\
    0 & 0 & 0 & 0 & 0 & 0 & | & 0 & 1 & 1 & 1 & 1 & 0\\
    \end{bmatrix}
\end{equation}
We perform the following operations to put the code into canonical form: swap rows $(4,5)$, $H$ on register $4$, swap registers $(5,6)$, then add row $4$ to rows $2$ and $3$, resulting in:
\begin{equation}
\setcounter{MaxMatrixCols}{19}
    \begin{bmatrix}
    1 & 0 & 0 & 0 & 0 & 1 & | & 1 & 0 & 1 & 1 & 0 & 1\\
    0 & 1 & 0 & 0 & 0 & 1 & | & 1 & 1 & 1 & 0 & 0 & 0\\
    0 & 0 & 1 & 0 & 0 & 0 & | & 1 & 1 & 1 & 1 & 0 & 0\\
    0 & 0 & 0 & 1 & 0 & 0 & | & 0 & 1 & 1 & 0 & 0 & 1\\
    0 & 0 & 0 & 0 & 1 & 0 & | & 0 & 0 & 0 & 0 & 0 & 0\\
    \end{bmatrix}
\end{equation}
\begin{equation}
    L=\begin{bmatrix}
    0 & 0 & 0 & 0 &0\\
    0 & 0 & 0 & 0 &0\\
    0 & 0 & 0 & 0 &0\\
    0 & -2 & 0 & 0 &0\\
    0 & 0 & 0 & 0 &0
    \end{bmatrix}
\end{equation}
\begin{equation}
\setcounter{MaxMatrixCols}{19}
    \begin{bmatrix}
    1 & 0 & 0 & 0 & 0 & 1 & | & 1 & 0 & 1 & 1 & 0 & 1\\
    0 & 1 & 0 & 0 & 0 & 1 & | & 1 & 1 & 1 & 0 & 0 & 0\\
    0 & 0 & 1 & 0 & 0 & 0 & | & 1 & 1 & 1 & 1 & 0 & 0\\
    0 & 0 & 0 & 1 & 0 & 0 & | & 0 & -1 & 1 & 0 & 0 & 1\\
    0 & 0 & 0 & 0 & 1 & 0 & | & 0 & 0 & 0 & 0 & 0 & 0\\
    \end{bmatrix}
\end{equation}
In this case $B=1$, and so the original bound based on random permutation of entries provides $p_D^*=8^4=4096$. This means that so long as more than $4096$ levels are used the distance of the code can be promised to at least remain the same. Below this value, unfortunately, one must apply numerical verification to ensure the distance is preserved.

\end{example}
\if{false}
\begin{example}
9-qubit shor code, I guess
\end{example}
\fi

Lastly, for the logical operators of the local-dimension-invariant representation of degenerate code the same argument holds as was given in \cite{gunderman2020local}. With all of these pieces we have an equally complete description of degenerate LDI codes as existed for non-degenerate LDI codes.

\section{Discussion}

The local-dimension-invariant (LDI) representation of stabilizer codes allows these codes to be applied regardless of the local-dimension of the underlying system. When introduced only non-degenerate codes could be written in local-dimension-invariant form and have their distance promised to be at least preserved, once the system had sufficiently many levels. In this work we have shown an alternative bound for how many levels are needed for the distance to be promised. While this bound suffers a severe dependency on the distance of the code, it does provide a nearly quadratic improvement on the dependency of the largest entry in the LDI form of the code, given by $B$. So while this bound is less helpful in some cases than the original bound it can be a tighter bound in others.

Beyond this, this work has shown that the LDI representation's associated distance promise also exists for degenerate quantum codes, and so completes the application of this technique to both families of standard stabilizer codes. Degenerate codes are of particular appeal since they are not restricted by the generalized quantum Hamming bound and can at times protect more logical particles than permitted by non-degenerate codes for a given distance and number of physical particles.

Unfortunately, the utility of this method is somewhat limited as both bounds on the required local-dimension are quite large. In order to improve the practicality of this technique the value for $p^*$ and $p_D^*$ must be significantly decreased. One way to reduce these bounds is to reduce the expression for $B$, the maximal entry in the LDI representation. To do so, other analysis techniques will be needed beyond simple counting arguments. Since the LDI form for a code is not unique, one possible method may be to solve systems of homogeneous linear diophantine equations, which given the surplus of variables (additions to entries) compared to variables (requirement of commutations to be zero) is likely to yield far smaller bounds on $B$. A starting point for this might include the following works: \cite{griffiths1946note,givens1947parametric}.

The results shown here extend the utility of local-dimension-invariant stabilizer codes, and so naturally there are questions as to what other uses this technique will have. Is it possible to apply this technique to show some foundational aspect of quantum measurements? Can this technique in some way be used for other varieties of stabilizer like codes, such as Entanglement-Assisted Quantum Error-Correcting Codes \cite{brun2006correcting,wilde2008optimal}? If this method can be applied in this situation it is possible that it could remove the need for entanglement use in these codes, so long as the local-dimension is altered. However, even still, the local-dimension required would likely be quite large so the importance of decreasing the bounds for $p^*$ and $p_D^*$ would become that much more.


\section*{Acknowledgments}

We thank Andrew Jena and David Cory for helpful comments.

\section*{Funding}
This work was supported by Industry Canada, the Canada
First Research Excellence Fund (CFREF), the Canadian Excellence Research Chairs (CERC 215284) program, the Natural Sciences and Engineering Research Council of Canada
(NSERC RGPIN-418579) Discovery program, the Canadian
Institute for Advanced Research (CIFAR), and the Province
of Ontario.

\if{false}

\section{Junk for this, but maybe useful for the future}

\section{ILP}

The problem can be stated as:
\begin{eqnarray}
    \min_{\{B_i\}} \max_{i,j} |B_i|_2^2 &&\\
    (\phi_q(s_i)+B_iq)\odot (\phi_q(s_j)+B_jq)&=&0\\
    B_i\in \mathbb{Z}^{2n}
\end{eqnarray}

\begin{definition}
$L^+=\phi(s_i)\odot \phi(s_j)$ pick the elements that are positive for $i<j$ and otherwise pick $j>i$. This provides the determinant theorem I want. That bound doesn't actually help...
\end{definition}

    $q$, $B-q^2d^2B$ gives $q^dd^d(B(q^2d^2+1))^d$ with ratio:
    
    \begin{eqnarray}
    & &\frac{(q-1)^{d-1}(d-1)^{d-1}(B((q-1)^2d^2+1))^{d-1}}{B^{2(d-1)}(2(d-1))^{(d-1)}}\\
    &=&\frac{(q-1)^{d-1}(((q-1)^2d^2+1))^{d-1}}{B^{(d-1)}2^{(d-1)}}\\
    &=&\left(\frac{(q-1)^2d^2+1}{2(2+k(q-1))}\right)^{d-1}
    \end{eqnarray}
this isn't always smaller than one, but often times is.

$qBd$ $q^3Bd^3+qBd$ gives $(qBd)^d(Bq^3d^3)^d$ waaay worse

\begin{proof}

Let us begin with a code over $q$ bases and extend it to $p$ bases. The errors for the original code are the vectors in the kernel of $\phi_q$ for the code. These errors are either unavoidable errors or are artifact errors. \if{false} We may rearrange the rows and columns so that the stabilizers and registers that generate these entries that are nonzero multiples of $q$ are the upper left $2d\times 2d$ minor, padding with identities if needed. The factor of 2 occurs due to the number of nonzero entries in $\phi_\infty$ being up to double the weight of the Pauli.\fi The stabilizer(s) that generate these multiples of $q$ entries in the syndrome are members of the null space of the minor formed using the corresponding stabilizer(s).

Now, consider the extension of the code to $p$ bases. Building up the qudit Pauli operators by weight $j$, we consider the minors of the matrix. These minors of size $2j\times 2j$ can have a nontrivial null space in two possible ways:
\begin{itemize}
    \item If the determinant is 0 over the integers then this is either an unavoidable error or an error whose existence did not occur due to the choice of the number of bases.
    \item If the determinant is not 0 over the integers, but takes the value of some multiple of $p$, then it's $0\mod p$ and so a null space exists.
\end{itemize}
Thus we can only introduce artifact errors to decrease the distance. By bounding the determinant by $p^*$, any choice of $p>p^*$ will ensure that the determinant is a unit in $\mathbb{Z}_p$, and hence have a trivial null space since the matrix is invertible.

We utilize the structure of the symplectic product more heavily in order to reduce the cutoff dimension. Note that for a pair of Paulis in the $\phi$ representation, we may write:
\begin{eqnarray}
    \phi(s_1)\odot \phi(s_2)&=&\phi(s_1)\begin{bmatrix} 0 & -I_n\\ I_n & 0 \end{bmatrix} \phi(s_2)^T\\
    &:=&\phi(s_1)g \phi(s_2)^T
\end{eqnarray}
and so we may consider the commutation for the generators with some Pauli $u$ as being given by $\bigoplus_{i=1}^{n-k} (\phi(s_i)g)\phi(u)^T$. Now, notice that for any Pauli weight $j$ operator, we will have up to $j$ nonzero entries in the $X$ component of the $\phi$ representation and up to $j$ nonzero entries in the $Z$ component. This means that up to $j$ columns in each component will be involved in any commutator, and so we may consider these $j$ columns alone for this operator and the generators will still be a direct sum: $[X_j\ |\ Z_j]g=[X_j\bigoplus Z_j]$.

Next note that to ensure that an artifact error is not induced it suffices to ensure that there is a nontrivial kernel, induced by the local-dimension choice, for $[X_j\bigoplus Z_j]$, which is ensured by (Kramers/Sylvesters rule) so long as any $2(d-1)\times 2(d-1)$ minor does not have a determinant which is congruent to the local-dimension. This can be promised by requiring the local dimension to be larger than the largest possible determinant for such a matrix: $p> det(X_j\bigoplus Z_j)$.

Through the partition theorem, we may write $X_j\bigoplus Z_j=(X_j\bigoplus I)(I\bigoplus Z_j)$ and so:
\begin{eqnarray}
|det([X_j\ |\ Z_j]g)|&=&|det(X_j\bigoplus Z_j)|\\
&=&|det(X_j)||det(Z_j)|.
\end{eqnarray} From here, we may apply Hadamard's inequality for determinants using the fact that the maximal entries are $q-1$ for $X_j$ and $B$ for $Z_j$, and each has dimension up to $(d-1)\times (d-1)$, which provides:
\begin{equation}
    p^*=((q-1)B)^{d-1}(d-1)^{d-1}
\end{equation}

\if{false}
Now, in order to guarantee that the value of $p$ is at least as large as the determinant, we can use Hadamard's inequality to obtain:
\begin{equation}
    p> p^* =B^{2(d-1)}(2(d-1))^{(d-1)}
\end{equation}
where $B$ is the maximal entry in $\phi_\infty$. Since we only need to ensure that the artifact induced null space is trivial for Paulis with weight less than $d$, we used this identity with $2(d-1)\times 2(d-1)$ matrices.
\fi

Lastly, when $j=d$, we can either encounter an unavoidable error, in which case the distance of the code is $d$ or we could obtain an artifact error, also causing the distance to be $d$. It is possible that neither of these occur at $j=d$, in which case the distance becomes some $d'$ with $d<d'\leq d^*$. 
\end{proof}

$Bd$ in Hadamard's inequality from $L^{(+)}$ since both halves are positive matrices provides $(Bd)^dd^{d/2}$.

\begin{theorem}\label{degen}
The distance promise can also be made for degenerate codes.
\end{theorem}

\begin{proof}
For the undetectable error case, this follows the same reasoning as the non-degenerate case of this theorem, so we only need to worry about two errors with the same syndrome mapping to different logical states.

For the degenerate syndrome error case, let $v$ and $u$ be two Pauli operators, with the same syndrome values, given by $\mathcal{J}$, with weight at most $d$. Let $\mathcal{A}$ be the (at most) $(d-1)$ registers such that $\mathcal{A}\cdot \phi(v)=\mathcal{J}$. Let $\mathcal{B}$ be the (at most) $(d-1)$ registers such that $\mathcal{B}\cdot \phi(u)=\mathcal{J}$. Generally $\phi(v)$ and $\phi(u)$ represent at least partly disjoint registers, however, if we apply a permutation $\tau$ on the registers in $u$ we may make the register locations seem to match: $\mathcal{B}\cdot \phi(u)=(\mathcal{B}\tau^{-1})\cdot (\tau\phi(u))$. Through this permutation $\tau\phi(u)$ and $\phi(v)$ represent the same registers. It suffices then to just prevent any linear combination of the columns in each of the $X$ and $Z$ components of $\mathcal{A}$ and $\mathcal{B}$ from having a non-trivial kernel due to the change in the local-dimension. We then require that:
\begin{equation}
    det(\mathcal{A}-\mathcal{B})<p,
\end{equation}
with $p$ being the new local-dimension. We may bound the determinant as before, although the maximal entries have been double, providing $p_D^*=(4(q-1)B(d-1))^{(d-1)}$.

While that suffices, it is not quite necessary. In many cases this value can be reduced. We may avoid introducing new degenerate syndrome values by ensuring that the local-dimension is larger than the largest possible syndrome value (over the integers), and so we also have $p_D^*=(d-1)(B+(q-1))$ as a valid bound. We may take the tighter of these two as our requirement to ensure the degeneracies do not cause an error which reduces the distance, meaning:
\begin{multline}p_D^*=\min\{(4(q-1)B(d-1))^{(d-1)},\\ (d-1)(B+(q-1))\}.\end{multline}
As $4B(q-1)\geq B+(q-1)$, the latter of these is always at most the size of the former, meaning that $p_D^*=(d-1)(B+(q-1))$ suffices.

In totality, we must have $p>\max\{p^*,p_D^*\}$ in order to ensure the distance is preserved.

\if{false}
It suffices, although is not necessary, to avoid allowing $(A-B\tau^{-1})$ to have a nontrivial nullspace that is introduced by changing the local-dimension. As in the undetectable error case, we may reduce this to bounding the determinant. In this case the maximal entry is up to twice as large, so we obtain $p_d^*=$.

For the degenerate syndrome error case, let $v$ and $u$ be two Pauli operators, with the same syndrome values, given by $\mathcal{J}$, with weight at most $d$. Let $A$ be the (at most) $(2d)\times (2d)$ matrix minor such that $A\cdot \phi(v)=\mathcal{J}$. Let $B$ be the (at most) $(2d)\times (2d)$ matrix minor in $C$ such that $B\cdot \phi(u)=\mathcal{J}$. $\phi(v)$ and $\phi(u)$ represent at least partly disjoint registers, however, if we apply a permutation $\tau$ on the registers in $u$ we may make the registers seem to match: $B\cdot \phi(u)=(B\tau^{-1})\cdot (\tau\phi(u))$. Through this permutation $\tau\phi(u)$ and $\phi(v)$ represent the same registers and have the same values (upon applying a scaling to the columns via a Clifford operation)--it suffices to just prevent any column combos. We then have that:
\begin{equation}
    (A-B\tau^{-1})\phi(v)=0.
\end{equation}
It suffices, although is not necessary, to avoid allowing $(A-B\tau^{-1})$ to have a nontrivial nullspace that is introduced by changing the local-dimension. As in the undetectable error case, we may reduce this to bounding the determinant. In this case the maximal entry is up to twice as large, so we obtain $p_d^*=$.

Also $p>(n-k-1)B+(n+k+1)(q-1)$ would suffice since the syndromes can only be that large at most.
\fi
\end{proof}

\begin{remark}
It's perfectly possible to turn a degenerate QECC into a non-degenerate one. Is the reverse possible? Yes, but avoided due to $p^*$ requirement, I think.
\end{remark}

\begin{corollary}
Given the bounds proven here, for all $q$ when $d\geq 3$ and for $q>2$ when $d\leq 2$ the same cutoff value is sufficient.
\end{corollary}

\begin{proof}
Consider $d=1$, in which case $p_D^*=p^*=0$.

When $d=2$, in order for the same cutoff value to suffice, we would need:
\begin{equation}
    (q-1)B\geq B+(q-1),
\end{equation}
which only holds if $q>2$.

For $d=3$ we have:
\begin{eqnarray}
    p^*&=&(d-1)^2(q-1)^2B^2\\
    &\geq& (d-1)[B^2+(q-1)^2+(d-3)(B+(q-1))]\\
    &\geq& (d-1)(B+(q-1))\\
    &=& p_D^*
\end{eqnarray}
Since this held for $d=3$, by monotonicity of exponentiation by a power greater than $1$ this will hold for all $d\geq 3$.

\end{proof}
\fi
\bibliographystyle{unsrt}
\phantomsection  
\renewcommand*{\bibname}{References}

\bibliography{main}

\end{document}